\documentclass[12pt]{article}
\usepackage{amsthm,amsmath,amssymb}
\usepackage{url}
\usepackage{enumerate}
\usepackage{subcaption}
\usepackage[utf8]{inputenc}
\usepackage[ruled,vlined]{algorithm2e}
\usepackage{xspace}
\usepackage{xcolor}
\usepackage{graphicx}
\usepackage[colorlinks=true,citecolor=black,linkcolor=black,urlcolor=blue]{hyperref}

\usepackage{xargs} 
\usepackage[colorinlistoftodos,prependcaption,textsize=tiny]{todonotes}

\def\acc#1{\left\{ #1 \right\}}

\theoremstyle{plain}
\newtheorem{theorem}{Theorem}
\newtheorem{lemma}[theorem]{Lemma}

\theoremstyle{definition}

\theoremstyle{remark}

\renewcommand{\le}{\leqslant}
\renewcommand{\ge}{\geqslant}

\title{Avoidability of formulas with two variables}

\author{Pascal Ochem\footnote{LIRMM, CNRS, Universit\'e de Montpellier, France. ochem@lirmm.fr}
and Matthieu Rosenfeld\footnote{LIP, ENS de Lyon, CNRS, UCBL, Universit\'e de Lyon, France. matthieu.rosenfeld@ens-lyon.fr}}
\begin{document}

\maketitle
\setcounter{footnote}{0}
\begin{abstract}
  In combinatorics on words, a word $w$ over an alphabet $\Sigma$ is
  said to avoid a pattern $p$ over an alphabet $\Delta$ of variables
  if there is no factor $f$ of $w$ such that $f=h(p)$ where $h:
  \Delta^*\to\Sigma^*$ is a non-erasing morphism. A pattern $p$ is
  said to be $k$-avoidable if there exists an infinite word over a
  $k$-letter alphabet that avoids $p$.
  We consider the patterns such that at most two variables appear at least twice,
  or equivalently, the formulas with at most two variables.
  For each such formula, we determine whether it is $2$-avoidable, and if it is $2$-avoidable,
  we determine whether it is avoided by exponentially many binary words.

\end{abstract}

%%%%%%%%%%%%%%%%%%%%%
\section{Introduction}\label{sec:intro}
%%%%%%%%%%%%%%%%%%%%%

A \emph{pattern} $p$ is a non-empty finite word over an alphabet
$\Delta=\acc{A,B,C,\dots}$ of capital letters called \emph{variables}.
An \emph{occurrence} of $p$ in a word $w$ is a non-erasing morphism $h:\Delta^*\to\Sigma^*$
such that $h(p)$ is a factor of $w$.
The \emph{avoidability index} $\lambda(p)$ of a pattern $p$ is the size of the
smallest alphabet $\Sigma$ such that there exists an infinite word
over $\Sigma$ containing no occurrence of $p$.
Bean, Ehrenfeucht, and McNulty~\cite{BEM79} and Zimin~\cite{Zimin}
characterized unavoidable patterns, i.e., such that $\lambda(p)=\infty$.
We say that a pattern $p$ is \emph{$t$-avoidable} if $\lambda(p)\le t$.
For more informations on pattern avoidability, we refer to Chapter 3 of Lothaire's book~\cite{Lothaire2002}.

A variable that appears only once in a pattern is said to be \emph{isolated}.
Following Cassaigne~\cite{Cassaigne1994}, we associate to a pattern $p$ the \emph{formula} $f$
obtained by replacing every isolated variable in $p$ by a dot.
The factors between the dots are called \emph{fragments}.

An \emph{occurrence} of $f$ in a word $w$ is a non-erasing morphism $h:\Delta^*\to\Sigma^*$
such that the $h$-image of every fragment of $f$ is a factor of $w$.
As for patterns, the avoidability index $\lambda(f)$ of a formula $f$ is the size of the
smallest alphabet allowing an infinite word containing no occurrence of $p$.
Clearly, every word avoiding $f$ also avoids $p$, so $\lambda(p)\le\lambda(f)$.
Recall that an infinite word is \emph{recurrent} if every finite factor appears
infinitely many times.
If there exists an infinite word over $\Sigma$ avoiding $p$,
then there there exists an infinite recurrent word over $\Sigma$ avoiding $p$.
This recurrent word also avoids $f$, so that $\lambda(p)=\lambda(f)$.
Without loss of generality, a formula is such that no variable is isolated
and no fragment is a factor of another fragment.

Cassaigne~\cite{Cassaigne1994} began and Ochem~\cite{Ochem2004}
finished the determination of the avoidability index of every pattern with at most 3 variables.
A \emph{doubled} pattern contains every variable at least twice.
Thus, a doubled pattern is a formula with exactly one fragment.
Every doubled pattern is 3-avoidable~\cite{O16}.
A formula is said to be \emph{binary} if it has at most 2 variables.
In this paper, we determine the avoidability index of every binary formula.

We say that a formula $f$ is \emph{divisible} by a formula $f'$ if $f$ does not avoid $f'$,
that is, there is a non-erasing morphism such that the image of any fragment of $f'$ by $h$ is a factor of a fragment of $f$.
If $f$ is divisible by $f'$, then every word avoiding $f'$ also avoids $f$ and thus $\lambda(f)\le\lambda(f')$.
Moreover, the reverse $f^R$ of a formula $f$ satisfies $\lambda(f^R)=\lambda(f)$.
For example, the fact that $ABA.AABB$ is 2-avoidable implies that $ABAABB$ and $BAB.AABB$ are 2-avoidable.
See Cassaigne~\cite{Cassaigne1994} and Clark~\cite{Clark} for more information on formulas and divisibility.
For convenience, we say that an avoidable formula $f$ is \emph{exponential} (resp. \emph{polynomial})
if the number of words in $\Sigma^n_{\lambda(f)}$ avoiding $f$ is exponential (resp. polynomial) in~$n$.

First, we check that every avoidable binary formula is 3-avoidable.
Since $\lambda(AA)=3$, every formula containing a square is 3-avoidable.
Then, the only square free avoidable binary formula is $ABA.BAB$ with avoidability index 3~\cite{Cassaigne1994}.
Thus, we have to distinguish between avoidable binary formulas with avoidability index 2 and 3.
A binary formula is minimally 2-avoidable if it is 2-avoidable and is not divisible by any other 2-avoidable binary formula.
A binary formula $f$ is maximally 2-unavoidable if it is 2-unavoidable and every other binary formula that is divisible by $f$ is 2-avoidable.

\begin{theorem}\label{main}{\ }

Up to symmetry, the maximally 2-unavoidable binary formulas are:
{\small
\begin{itemize}
\item $AAB.ABA.ABB.BBA.BAB.BAA$
\item $AAB.ABBA$
\item $AAB.BBAB$
\item $AAB.BBAA$
\item $AAB.BABB$
\item $AAB.BABAA$
\item $ABA.ABBA$
\item $AABA.BAAB$
\end{itemize}
}
Up to symmetry, the minimally 2-avoidable binary formulas are:
\begin{itemize}
{\small
\item $AA.ABA.ABBA$ (polynomial)
\item $ABA.AABB$ (polynomial)
\item $AABA.ABB.BBA$ (polynomial)
\item $AA.ABA.BABB$ (exponential)
\item $AA.ABB.BBAB$ (exponential)
\item $AA.ABAB.BB$ (exponential)
\item $AA.ABBA.BAB$ (exponential)
\item $AAB.ABB.BBAA$ (exponential)
\item $AAB.ABBA.BAA$ (exponential)
\item $AABB.ABBA$ (exponential)
\item $ABAB.BABA$ (exponential)
\item $AABA.BABA$ (exponential)
\item $AAA$ (exponential)
\item $ABA.BAAB.BAB$ (exponential)
\item $AABA.ABAA.BAB$ (exponential)
\item $AABA.ABAA.BAAB$ (exponential)
\item $ABAAB$ (exponential)
}
\end{itemize}
\end{theorem}

Given a binary formula $f$, we can use Theorem~\ref{main} to find $\lambda(f)$.
Now, we also consider the problem whether an avoidable binary formula is polynomial or exponential.
If $\lambda(f)=3$, then either $f$ contains a square or $f=ABA.BAB$, so that $f$ is exponential.
Thus, we consider only the case $\lambda(f)=2$.
If $f$ is divisible by an exponential $2$-avoidable formula given in Theorem~\ref{main}, then $f$ is known to be exponential.
This leaves open the case such that $f$ is only divisible by polynomial $2$-avoidable formulas.
The next result settles every open case.

\begin{theorem}\label{second}{\ }

The following formulas are polynomial:
{\small
\begin{itemize}
\item $BBA.ABA.AABB$
\item $AABA.AABB$
\end{itemize}
}
The following formulas are exponential:
\begin{itemize}
{\small
\item $BAB.ABA.AABB$
\item $AAB.ABA.ABBA$
\item $BAA.ABA.AABB$
\item $BBA.AABA.AABB$
}
\end{itemize}
\end{theorem}

To obtain the 2-unavoidability of the formulas in the first part of Theorem~\ref{main},
we use a standard backtracking algorithm.
Figure~\ref{tabularunavoidable} gives the maximal length and number of binary words avoiding each maximally 2-unavoidable formula.

\begin{figure}[htbp]
{
\small
\begin{tabular}{|c|c|c|}
\hline
 & Maximal length of a & Number of binary\\
Formula & binary word avoiding & words avoiding\\
 & this formula & this formula\\
\hline
 $AAB.BBAA$& 22 & 1428\\
\hline
 $AAB.ABA.ABB.BBA.BAB.BAA$& 23 & 810\\
\hline
 $AAB.BBAB$& 23 & 1662\\
\hline
 $AABA.BAAB$& 26 & 2124\\
\hline
 $AAB.ABBA$& 30 & 1684\\
\hline
$AAB.BABAA$&42 & 71002\\
\hline
$AAB.BABB$& 69 & 9252\\
\hline
$ABA.ABBA$ & 90 & 31572 \\
\hline
\end{tabular}
}
\caption{The number and maximal length of binary words avoiding the maximally 2-unavoidable formulas.}
\label{tabularunavoidable}
\end{figure}

In Section~\ref{sec:poly}, we consider the polynomial formulas in Theorems~\ref{main} and~\ref{second}.
The proof uses a technical lemma given in Section~\ref{sec:lemma}.
Then we consider in Section~\ref{sec:exp} the exponential formulas in Theorems~\ref{main} and~\ref{second}.

A preliminary version of this paper, without Theorem~\ref{second}, has been presented at DLT 2016.

%%%%%%%%%%%%%%%%%%%%%
\section{The useful lemma}\label{sec:lemma}
%%%%%%%%%%%%%%%%%%%%%

Let us define the following words:
\begin{itemize}
\item$b_2$ is the fixed point of $\texttt{0}\mapsto\texttt{01}$, $\texttt{1}\mapsto\texttt{10}$.
\item$b_3$ is the fixed point of $\texttt{0}\mapsto\texttt{012}$, $\texttt{1}\mapsto\texttt{02}$, $\texttt{2}\mapsto\texttt{1}$.
\item$b_4$ is the fixed point of $\texttt{0}\mapsto\texttt{01}$, $\texttt{1}\mapsto\texttt{03}$, $\texttt{2}\mapsto\texttt{21}$, $\texttt{3}\mapsto\texttt{23}$.
\item$b_5$ is the fixed point of $\texttt{0}\mapsto\texttt{01}$, $\texttt{1}\mapsto\texttt{23}$, $\texttt{2}\mapsto\texttt{4}$, $\texttt{3}\mapsto\texttt{21}$, $\texttt{4}\mapsto\texttt{0}$.
\end{itemize}

Let $w$ and $w'$ be infinite (right infinite or bi-infinite) words. We say that $w$ and $w'$ are equivalent if they have the same set of finite factors.
We write $w\sim w'$ if $w$ and $w'$ are equivalent.
%Thue~\cite{Thue06} proved that $b_3$ characterizes the ternary square free words avoiding the factors \texttt{010} and \texttt{212}.
A famous result of Thue~\cite{Thue06} can be stated as follows:

\begin{theorem}~\cite{Thue06}
\label{thm:thue}
%$b_3$ is the only bi-infinite ternary word avoiding \texttt{010}, \texttt{212}, and squares.
Every bi-infinite ternary word avoiding \texttt{010}, \texttt{212}, and squares is equivalent to $b_3$.
\end{theorem}

Given an alphabet $\Sigma$ and forbidden structures $S$, we say that a finite set $W$ of infinite words over $\Sigma$
\emph{essentially avoids} $S$ if every word in $W$ avoids $S$ and every bi-infinite words over $\Sigma$ avoiding $S$
is equivalent to one of the words in $S$. If $W$ contains only one word $w$, we denote the set $W$ by $w$ instead of $\acc{w}$.
Then we can restate Theorem~\ref{thm:thue}: $b_3$ essentially avoids \texttt{010}, \texttt{212}, and squares

The results in the next section involve $b_3$.
We have tried without success to prove them by using Theorem~\ref{thm:thue}.
We need the following stronger property of $b_3$:
\begin{lemma}
\label{lm}
$b_3$ essentially avoids \texttt{010}, \texttt{212}, $XX$ with $1\le|X|\le3$, and $\texttt{2}YY$ with $|Y|\ge4$.
\end{lemma}

\begin{proof}
We start by checking by computer that $b_3$ has the same set of factors of length 100 as
every bi-infinite ternary word avoiding \texttt{010}, \texttt{212}, $XX$ with $1\le|X|\le3$, and $\texttt{2}YY$ with $|Y|\ge4$.
The set of the forbidden factors of $b_3$ of length at most 4 is
$F=\acc{\texttt{00},\texttt{11},\texttt{22},\texttt{010},\texttt{212},\texttt{0202},\texttt{2020},\texttt{1021},\texttt{1201}}$.
To finish the proof, we use Theorem~\ref{thm:thue} and we suppose for contradiction that $w$ is a bi-infinite ternary word that
contains a large square $MM$ and avoids both $F$ and large factors of the form $\texttt{2}YY$.
%\unsure{On peut préciser que la contradiction vient du résultat de Thue que tu cite au dessus?}

\begin{itemize}
 \item Case $M=\texttt{0}N$. Then $w$ contains $MM=\texttt{0}N\texttt{0}N$. Since $\texttt{00}\in F$ and $\texttt{2}YY$ is forbidden, $w$ contains $\texttt{10}N\texttt{0}N$.
 Since $\acc{\texttt{11},\texttt{010}}\subset F$, $w$ contains $210N0N$.
 If $N=P\texttt{1}$, then $w$ contains $\texttt{210}P\texttt{10}P\texttt{1}$, which contains $\texttt{2}YY$ with $Y=10P$.
 So $N=P\texttt{2}$ and $w$ contains $\texttt{210}P\texttt{20}P\texttt{2}$.
 If $P=Q\texttt{1}$, then $w$ contains $\texttt{210}Q\texttt{120}Q\texttt{12}$.
 Since $\acc{\texttt{11},\texttt{212}}\subset F$, the factor $Q\texttt{12}$ implies that $Q=R\texttt{0}$ and $w$ contains $\texttt{210}R\texttt{0120}R\texttt{012}$.
 Moreover, since $\acc{\texttt{00},\texttt{1201}}\subset F$, the factor $\texttt{120}R$ implies that $R=\texttt{2}S$ and $w$ contains $\texttt{2102}S\texttt{01202}S\texttt{012}$.
 Then there is no possible prefix letter for $S$: \texttt{0} gives \texttt{2020}, \texttt{1} gives \texttt{1021}, and \texttt{2} gives \texttt{22}.
 This rules out the case $P=Q\texttt{1}$.
 So $P=Q\texttt{0}$ and $w$ contains $\texttt{210}Q\texttt{020}Q\texttt{02}$.
 The factor $Q\texttt{020}Q$ implies that $Q=1R1$, so that $w$ contains $\texttt{2101}R\texttt{10201}R\texttt{102}$.
 Since $\acc{\texttt{11},\texttt{010}}\subset F$, the factor $\texttt{01}R$ implies that $R=\texttt{2}S$, so that $w$ contains $\texttt{21012}S\texttt{102012}S\texttt{102}$.
 The only possible right extension with respect to $F$ of \texttt{102} is \texttt{102012}.
 So $w$ contains $\texttt{21012}S\texttt{102012}S\texttt{102012}$, which contains $\texttt{2}YY$ with $Y=S\texttt{102012}$.
 \item Case $M=\texttt{1}N$. Then $w$ contains $MM=\texttt{1}N\texttt{1}N$. In order to avoid $\texttt{11}$ and $\texttt{2}YY$, $w$ must contain $\texttt{01}N\texttt{1}N$.
 If $N=P\texttt{0}$, then $w$ contains $\texttt{01}P\texttt{01}P\texttt{0}$. So $w$ contains the large square $\texttt{01}P\texttt{01}P$ and this case is covered by the previous item.
 So $N=P\texttt{2}$ and $w$ contains $\texttt{01}P\texttt{21}P\texttt{2}$. 
 Then there is no possible prefix letter for $P$: \texttt{0} gives \texttt{010}, \texttt{1} gives \texttt{11}, and \texttt{2} gives \texttt{212}.
 \item Case $M=\texttt{2}N$. Then $w$ contains $MM=2N2N$.
 If $N=P\texttt{1}$, then $w$ contains $\texttt{2}P\texttt{12}P\texttt{1}$. This factor cannot extend to $\texttt{2}P\texttt{12}P\texttt{12}$, since this is $\texttt{2}YY$ with $Y=P\texttt{12}$.
 So $w$ contains $\texttt{2}P\texttt{12}P\texttt{10}$. Then there is no possible suffix letter for $P$: \texttt{0} gives \texttt{010}, \texttt{1} gives \texttt{11}, and \texttt{2} gives \texttt{212}.
 This rules out the case $N=P\texttt{1}$.
 So $N=P\texttt{0}$ and $w$ contains $\texttt{2}P\texttt{02}P\texttt{0}$.
 This factor cannot extend to $\texttt{02}P\texttt{02}P\texttt{0}$, since this contains the large square $\texttt{02}P\texttt{02}P$ and this case is covered by the first item.
 Thus $w$ contains $\texttt{12}P\texttt{02}P\texttt{0}$.
 If $P=Q\texttt{1}$, then $w$ contains $\texttt{12}Q\texttt{102}Q\texttt{10}$.
 Since $\acc{\texttt{22},\texttt{1021}}\subset F$, the factor $\texttt{102}Q$ implies that $Q=\texttt{0}R$, so that $w$ contains $\texttt{120}R\texttt{1020}R\texttt{10}$.
 Then there is no possible prefix letter for $R$: \texttt{0} gives \texttt{00}, \texttt{1} gives \texttt{1201}, and \texttt{2} gives \texttt{0202}.
 This rules out the case $P=Q\texttt{1}$.
 So $P=Q\texttt{2}$ and $w$ contains $\texttt{12}Q\texttt{202}Q\texttt{20}$. The factor $Q\texttt{202}$ implies that $Q=R1$ and $w$ contains $\texttt{12}R\texttt{1202}R\texttt{120}$.
 Since $\acc{\texttt{00},\texttt{1201}}\subset F$, $w$ contains $\texttt{12}R\texttt{1202}R\texttt{1202}$, which contains $\texttt{2}YY$ with $Y=R\texttt{1202}$.
\end{itemize}

\end{proof}

% The proof of Lemma~\ref{lm} can be found in the full version of the paper at \arxiv{1606.03955}.

%%%%%%%%%%%%%%%%%%%%%
\section{Polynomial formulas}\label{sec:poly}
%%%%%%%%%%%%%%%%%%%%%

% Let us consider the 2-avoidable formulas in the second part of Theorem~\ref{main}.
% Let $b_3$  denote the bi-infinite word with the same set of factors as the fixed point of
% $\texttt{0}\mapsto\texttt{012}$, $\texttt{1}\mapsto\texttt{02}$, $\texttt{2}\mapsto\texttt{1}$,
%\marginpar{J'ai viré les 3 autres puisque là on ne s'en sert pas.}
Let us detail the binary words avoiding the polynomial formulas in Theorems~\ref{main} and~\ref{second}.
\begin{lemma}\label{lem:poly}{\ }
\begin{itemize}
\item $\acc{g_x(b_3), g_y(b_3), g_z(b_3), g_{\overline{z}}(b_3)}$ essentially avoids $AA.ABA.ABBA$. 
\item $g_x(b_3)$ essentially avoids $AABA.ABB.BBA$.
\item Let $f$ be either $ABA.AABB$, $BBA.ABA.AABB$, or $AABA.AABB$. Then $\acc{g_x(b_3), g_t(b_3)}$ essentially avoids $f$.
\end{itemize}
\end{lemma}

% The first three formulas do not seem to be avoided by exponentially many binary words.
% \begin{itemize}
% \item $AA.ABA.ABBA$ is avoided by $g_x(b_3)$, $g_y(b_3)$, $g_z(b_3)$, and $g_{\overline{z}}(b_3)$.
% \item $ABA.AABB$ is avoided by $g_x(b_3)$ and $g_t(b_3)$.
% \item $AABA.ABB.BBA$ is avoided by $g_x(b_3)$.
% \end{itemize}
% We conjecture that the only bi-infinite recurrent words avoiding these formulas are the ones given above.
The words avoiding these formulas are morphic images of $b_3$ by the morphisms given below.
Let $\overline{w}$ denote the word obtained from the (finite or bi-infinite) binary word $w$ by exchanging \texttt{0} and \texttt{1}.
Obviously, if $w$ avoids a given formula, then so does $\overline{w}$.
A (bi-infinite) binary word $w$ is \emph{self-complementary} if $w\sim\overline{w}$.
The words $g_x(b_3)$, $g_y(b_3)$, and $g_t(b_3)$ are self-complementary.
Since the frequency of \texttt{0} in $g_z(b_3)$ is $\tfrac59$, $g_z(b_3)$ is not self-complementary.
Then $g_{\overline{z}}$ is obtained from $g_z$ by exchanging \texttt{0} and \texttt{1}, so that $g_{\overline{z}}(b_3)=\overline{g_z(b_3)}$.
%\unsure{Tu ne veux pas préciser qu'on ne parle que de $g_{\overline{z}}$ parce que les autres donnent le même mot bi-infini?}

%($g_{\overline{z}}(b_3)$ does not have the same set of factors than $g_z(b_3)$, this is not the case for $g_x$ or $g_y$).

\noindent
\begin{minipage}[b]{0.23\linewidth}
\centering
$$\begin{array}{l}
 g_x(\texttt{0})=\texttt{01110},\\
 g_x(\texttt{1})=\texttt{0110},\\
 g_x(\texttt{2})=\texttt{0}.\\ 
\end{array}$$
\end{minipage}
\begin{minipage}[b]{0.21\linewidth}
\centering
$$\begin{array}{l}
 g_y(\texttt{0})=\texttt{0111},\\
 g_y(\texttt{1})=\texttt{01},\\
 g_y(\texttt{2})=\texttt{00}.\\ 
\end{array}$$
\end{minipage}
\begin{minipage}[b]{0.21\linewidth}
\centering
$$\begin{array}{l}
 g_z(\texttt{0})=\texttt{0001},\\
 g_z(\texttt{1})=\texttt{001},\\
 g_z(\texttt{2})=\texttt{11}.\\ 
\end{array}$$
\end{minipage}
\begin{minipage}[b]{0.25\linewidth}
\centering
$$\begin{array}{l}
 g_t(\texttt{0})=\texttt{01011011010},\\
 g_t(\texttt{1})=\texttt{01011010},\\
 g_t(\texttt{2})=\texttt{010}.\\ 
\end{array}$$
\end{minipage}
\\

Let us first state interesting properties of the morphisms and the formulas in Lemma~\ref{lem:poly}.
\begin{lemma}\label{lem:occ}
For every $p,s\in\Sigma_3$, $Y\in\Sigma_3^*$ with $|Y|\ge4$, and $g\in\acc{g_x,g_y,g_z,g_{\overline{z}},g_t}$,
the word $g(p2YYs)$ contains an occurrence of $AABA.AABBA$. 
\end{lemma}

\begin{proof}{\ }
\begin{itemize}
 \item Since \texttt{0} is a prefix and a suffix of the $g_x$-image of every letter, $g_x(p\texttt{2}YYs)=V\texttt{000}U\texttt{00}U\texttt{00}W$
 contains an occurrence of $AABA.AABBA$ with $A=\texttt{0}$ and $B=\texttt{0}U\texttt{0}$.
 \item Since \texttt{0} is a prefix of the $g_y$-image of every letter, $g_y(\texttt{2}YYs)=\texttt{000}U\texttt{0}U\texttt{0}V$ with $U,V\in\Sigma^+_3$,
 which contains an occurrence of $AABA.AABBA$ with $A=\texttt{0}$ and $B=\texttt{0}U$.
 \item Since \texttt{1} is a suffix of the $g_z$-image of every letter, $g_z(p\texttt{2}YY)=\texttt{111}U\texttt{1}U\texttt{1}$
 contains an occurrence of $AABA.AABBA$ with $A=\texttt{1}$ and $B=\texttt{1}U$.
 \item Since $g_{\overline{z}}(p\texttt{2}YY)=\overline{g_z(p\texttt{2}YY)}$, $g_{\overline{z}}(s\texttt{2}YY)$ contains an occurrence of $AABA.AABBA$.
 \item Since \texttt{010} is a prefix and a suffix of the $g_t$-image of every letter, $g_t(p\texttt{2}YYs)=V\texttt{010010010}U\texttt{010010}U\texttt{010010}W$
 contains an occurrence of $AABA.AABBA$ with $A=\texttt{010}$ and $B=\texttt{010}U\texttt{010}$.
\end{itemize}
\end{proof}

\begin{lemma}\label{rem}
$AABA.AABBA$ is divisible by every formula in Lemma~\ref{lem:poly}.
\end{lemma}

We are now ready to prove Lemma~\ref{lem:poly}.
To prove the avoidability, we have implemented Cassaigne's algorithm that decides, under mild assumptions,
whether a morphic word avoids a formula~\cite{Cassaigne1994}.
We have to explain how the long enough binary words avoiding a formula can be split into 4 or 2 distinct
incompatible types. A similar phenomenon has been described for $AABB.ABBA$~\cite{aabbc}.

First, consider any infinite binary word $w$ avoiding $AA.ABA.ABBA$.
A computer check shows by backtracking that $w$ must contain
the factor \texttt{01110001110}. In particular, $w$ contains \texttt{00}.
Thus, $w$ cannot contain both \texttt{010} and \texttt{0110}, since it would produce an occurrence of $AA.ABA.ABBA$.
Moreover, a computer check shows by backtracking that $w$ cannot avoid both \texttt{010} and \texttt{0110}.
So, $w$ must contain either \texttt{010} or \texttt{0110} (this is an exclusive or).
By symmetry, $w$ must contain either \texttt{101} or \texttt{1001}.
There are thus at most 4 possibilities for $w$, depending on which subset of $\acc{\texttt{010},\texttt{0110},\texttt{101},\texttt{1001}}$
appears among the factors of $w$, see Figure~\ref{fig4}.

\begin{figure}[htbp]
\centering
\includegraphics[width=4.8cm]{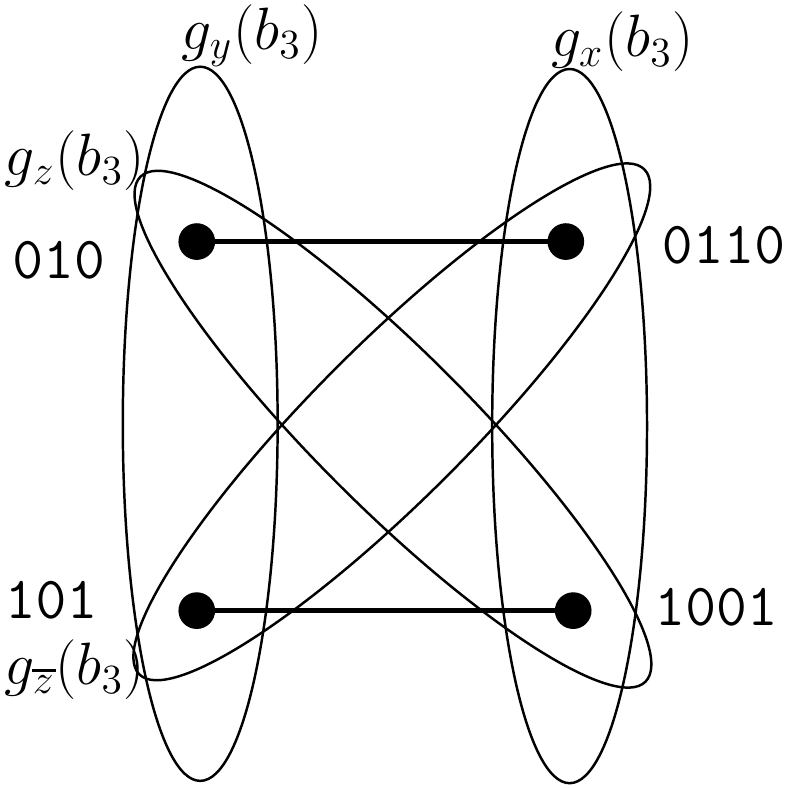}
\caption{The four infinite binary words avoiding $AA.ABA.ABBA$.}
\label{fig4}
\end{figure}

% \begin{figure}[htbp]
% \begin{subfigure}{.4\linewidth}
% \centering
% \includegraphics[width=4.8cm]{type4}
% \caption{The four bi-infinite binary words avoiding $AA.ABA.ABBA$.}
% \label{fig4}
% \end{subfigure}\hfill
% \begin{subfigure}{.5\linewidth}
% \centering
% \includegraphics[width=4.8cm]{type2}
% \caption{The two bi-infinite binary words avoiding $ABA.AABB$.}
% \label{fig2}
% \end{subfigure}
% \caption{}
% \end{figure}

Also, consider any infinite binary word $w$ avoiding $f$, where $f$ is either $ABA.AABB$, $BBA.ABA.AABB$, or $AABA.AABB$.
Notice that the formulas $BBA.ABA.AABB$ and $AABA.AABB$ are divisible by $ABA.AABB$.
%First, we consider the polynomial formulas. Let $f$ be either $BBA.ABA.AABB$ or $AABA.AABB$.
%Using the technique described in Section~\ref{sec:poly}, we show that $\acc{g_x(b_3), g_t(b_3)}$ essentially avoids $f$.
%Since $f$ is divisible by $ABA.AABB$, this improves our previous result that $\acc{g_x(b_3), g_t(b_3)}$ essentially avoids $ABA.AABB$.
We check by backtracking that no infinite binary word avoids $f$, \texttt{0010}, and \texttt{00110}.
A word containing both \texttt{0010} and \texttt{00110} contains an occurrence of $AABA.AABBA$, and thus an occurrence of $f$ by Lemma~\ref{rem}.
So $w$ does not contain both \texttt{0010} and \texttt{00110}.
Thus, there are two possibilities for $w$ depending on whether it contains \texttt{0010} or \texttt{00110}.
% 
% Clearly, an infinite binary word avoiding $f$ does not contain both \texttt{00101101}, and \texttt{00110011}.
% We check that the set of prolongable binary words of length 100 avoiding $f$ \texttt{00101101}
% is exactly the set of factors of length 100 of $g_x(b_3)$.
% We have seen that $g_x(u2YYa)$ contains an occurrence of $AABA.AABBA$, which divides $f$.
% Thus, 
% Also, the set of prolongable binary words of length 100 avoiding $f$ \texttt{00110011}
% is exactly the set of factors of length 100 of $g_t(b_3)$.
% 
% Notice that $w$ cannot contain both \texttt{010} and \texttt{0011}.
% Also, a computer check shows by backtracking that $w$ cannot avoid both \texttt{010} and \texttt{1100}.
% By symmetry, there are thus at most 2 possibilities for $w$, depending on which subset of
% $\acc{\texttt{010},\texttt{0011},\texttt{101},\texttt{1100}}$ appears among the factors of $w$, see Figure~\ref{fig2}.

Now, our tasks of the form "prove that a set of morphic words essentially avoids one formula" are reduced to (more) tasks of the form
"prove that one morphic word essentially avoids one formula and a set of factors".

Since all the proofs of such reduced tasks are very similar,
we only detail the proof that $g_y(b_3)$ essentially avoids $AA.ABA.ABBA$, \texttt{0110}, and \texttt{1001}.
%As an example
We check that the set of prolongable binary words of length $100$ avoiding $AA.ABA.ABBA$, \texttt{0110}, and \texttt{1001}
is exactly the set of factors of length $100$ of $g_y(b_3)$.
%Following , we say that a morphism $g$ is \emph{circular} for a factorial language $\mathcal{L}$ if for every sufficiently long factor $u$
%of every word $g(w)$ with $w\in\mathcal{L}$, we have
Using Cassaigne's notion of circular morphism~\cite{Cassaigne1994}, this is sufficient to prove that every bi-infinite binary word
of this type is the $g_y$-image of some bi-infinite ternary word $w_3$.
%\unsure{Implicitement tu utilise le fait que $g_y$ est synchronisant? }
It also ensures that $w_3$ and $b_3$ have the same set of small factors.
%\unsure{Mais du coups la longueur nécessaire c'est genre $<(2+6)*4= 24 << 500$. Donc 500 c'est juste pour avoir de la marge? Le 500 fait un peu peur.}
Suppose for contradiction that $w_3\not\sim b_3$.
By Lemma~\ref{lm}, $w_3$ contains a factor $\texttt{2}YY$ with $|Y|\ge4$.
Since $w_3$ is bi-infinite, $w_3$ even contains a factor $p\texttt{2}YYs$ with $p,s\in\Sigma_3$.
By Lemma~\ref{lem:occ}, $g_y(w_3)$ contains an occurrence of $AABA.AABBA$ and by Lemma~\ref{rem}, $g_y(w_3)$ contains an occurrence of $AA.ABA.ABBA$.
This contradiction shows that $w_3\sim b_3$. So $g_y(b_3)$ essentially avoids $AA.ABA.ABBA$, \texttt{0110}, and \texttt{1001}.

% Thus, the proof 
% 
% We plan to prove that the mentioned words avoiding the formulas $AA.ABA.ABBA$, $ABA.AABB$,
% and $AABA.ABB.BBA$ are indeed the only ones. Similar results have been obtained in~\cite{BO15}.
% However, the proof technique heavily relies on that fact that the considered words avoid large squares,
% whereas our words contain arbitrarily large squares.

%%%%%%%%%%%%%%%%%%%%%
\section{Exponential formulas}\label{sec:exp}
%%%%%%%%%%%%%%%%%%%%%
% \marginpar{On ne dit pas que on fait une autre preuve parce que la technique cassaigne marche pas bien si c'est trop gros? 
% (Bon d'un autre côté ta preuve là à l'avantage qu'on peut presque le faire à la main donc moi ça suffit à me convaincre que c'est util de la donner)}

Given a morphism $g:\Sigma^*_3\to\Sigma^*_2$, an sqf-$g$-image
is the image by $g$ of a (finite or infinite) ternary square free word.
With an abuse of language, we say that $g$ avoids a set of formulas if
every sqf-$g$-image avoids every formula in the set.
For every 2-avoidable exponential formula $f$ in Theorems~\ref{main} and~\ref{second},
we give below a uniform morphism $g$ that avoids $f$.
%For every $q$-uniform morphism $g$ above, we say that a binary word is an sqf-$g$-image if it is the $g$-image of .
If possible, we simultaneously avoid the reverse formula $f^R$ of $f$.
We also avoid large squares.
Let $SQ_t$ denote the pattern corresponding to squares of period at least $t$,
that is, $SQ_1=AA$, $SQ_2=ABAB$, $SQ_3=ABCABC$, and so on.
The morphism $g$ avoids $SQ_t$ with $t$ as small as possible.
Since $\lambda(SQ_2)$, a binary word avoiding $SQ_3$ is necessarily best possible in terms of length of avoided squares.

\begin{itemize}
\item $f=AA.ABA.BABB$. This $22$-uniform morphism avoids $\acc{f,f^R,SQ_6}$:
{\small
$$
\begin{array}{c}
\texttt{0}\mapsto\texttt{0001101101110011100011}\\
\texttt{1}\mapsto\texttt{0001101101110001100011}\\
\texttt{2}\mapsto\texttt{0001101101100011100111}\\
\end{array}
$$
}
This $44$-uniform morphism avoids $\acc{f,SQ_5}$:
{\small
$$
\begin{array}{c}
\texttt{0}\mapsto\texttt{00010010011000111001001100010011100100100111}\\
\texttt{1}\mapsto\texttt{00010010011000100111001001100011100100100111}\\
\texttt{2}\mapsto\texttt{00010010011000100111001001001100011100100111}\\
\end{array}
$$
}
Notice that $\acc{f,f^R,SQ_5}$ is 2-unavoidable and $\acc{f,SQ_4}$ is 2-unavoidable.
\item $f=AA.ABB.BBAB$. This $60$-uniform morphism avoids $\acc{f,f^R,SQ_{11}}$:
{\small
$$
\begin{array}{c}
\texttt{0}\mapsto\texttt{000110011100011001110011000111000110011100011100110001110011}\\
\texttt{1}\mapsto\texttt{000110011100011001110001110011000111000110011100110001110011}\\
\texttt{2}\mapsto\texttt{000110011100011001110001100111000111001100011100110001110011}\\
\end{array}
$$
}
Notice that $\acc{f,SQ_{10}}$ is 2-unavoidable.
\item $f=AA.ABAB.BB$ is self-reverse. This $11$-uniform morphism avoids $\acc{f,SQ_4}$:
{\small
$$
\begin{array}{c}
\texttt{0}\mapsto\texttt{00100110111}\\
\texttt{1}\mapsto\texttt{00100110001}\\
\texttt{2}\mapsto\texttt{00100011011}\\
\end{array}
$$
}
Notice that $\acc{f,SQ_3}$ is 2-unavoidable.
\item $f=AA.ABBA.BAB$ is self-reverse. This $30$-uniform morphism avoids $\acc{f,SQ_6}$:
{\small
$$
\begin{array}{c}
\texttt{0}\mapsto\texttt{000110001110011000110011100111}\\
\texttt{1}\mapsto\texttt{000110001100111001100011100111}\\
\texttt{2}\mapsto\texttt{000110001100011001110011100111}\\
\end{array}
$$
}
Notice that $\acc{f,SQ_5}$ is 2-unavoidable.
\item $f=AAB.ABB.BBAA$ is self-reverse. This $30$-uniform morphism avoids $\acc{f,SQ_5}$:
{\small
$$
\begin{array}{c}
\texttt{0}\mapsto\texttt{000100101110100010110111011101}\\
\texttt{1}\mapsto\texttt{000100101101110100010111011101}\\
\texttt{2}\mapsto\texttt{000100010001011101110111010001}\\
\end{array}
$$
}
Notice that $\acc{f,SQ_4}$ is 2-unavoidable.
\item $f=AAB.ABBA.BAA$ is self-reverse. This $38$-uniform morphism avoids $\acc{f,SQ_5}$:
{\small
$$
\begin{array}{c}
\texttt{0}\mapsto\texttt{00010001000101110111010001011100011101}\\
\texttt{1}\mapsto\texttt{00010001000101110100011100010111011101}\\
\texttt{2}\mapsto\texttt{00010001000101110001110100010111011101}\\
\end{array}
$$
}
Notice that $\acc{f,SQ_4}$ is 2-unavoidable.
\item $f=AABB.ABBA$. This $193$-uniform morphism avoids $\acc{f,SQ_{16}}$:
{\small
$$
\begin{array}{l}
\texttt{0}\mapsto\texttt{00010001011011101100010110111000101101110111000101100010001011}\\
\texttt{011101100010110111011100010110111011000101101110001011011101110001}\\
\texttt{01100010001011011100010110111011100010110111011000101101110001011}\\
\texttt{1}\mapsto\texttt{00010001011011101100010110111000101101110111000101100010001011}\\
\texttt{011100010110111011100010110111011000101101110001011011101110001011}\\
\texttt{00010001011011101100010110111011100010110111011000101101110001011}\\
\texttt{2}\mapsto\texttt{00010001011011100010110111011100010110001000101101110110001011}\\
\texttt{011101110001011011101100010110111000101101110111000101100010001011}\\
\texttt{01110110001011011100010110111011100010110111011000101101110001011}\\
\end{array}
$$
}
Notice that $\acc{f,f^R}$ is 2-unavoidable and $\acc{f,SQ_{15}}$ is 2-unavoidable.
Previous papers~\cite{Ochem2004,aabbc} have considered a $102$-uniform morphism to avoid $\acc{f,SQ_{27}}$.
%0001000101101110110001011011100010110111011100010110001000101101110110001011011101110001011011101100010110111000101101110111000101100010001011011100010110111011100010110111011000101101110001011
%0001000101101110110001011011100010110111011100010110001000101101110001011011101110001011011101100010110111000101101110111000101100010001011011101100010110111011100010110111011000101101110001011
%0001000101101110001011011101110001011000100010110111011000101101110111000101101110110001011011100010110111011100010110001000101101110110001011011100010110111011100010110111011000101101110001011

\item $f=ABAB.BABA$ is self-reverse. This $50$-uniform morphism avoids $\acc{f,SQ_3}$, see~\cite{Ochem2004}:
{\small
$$
\begin{array}{c}
\texttt{0}\mapsto\texttt{00011001011000111001011001110001011100101100010111}\\
\texttt{1}\mapsto\texttt{00011001011000101110010110011100010110001110010111}\\
\texttt{2}\mapsto\texttt{00011001011000101110010110001110010111000101100111}\\
\end{array}
$$
}
Notice that a binary word avoiding $\acc{f,SQ_3}$ contains only the squares
$\texttt{00}$, $\texttt{11}$, and $\texttt{0101}$ (or $\texttt{00}$, $\texttt{11}$, and $\texttt{1010}$).
\item $f=AABA.BABA$: A case analysis of the small factors shows that a recurrent binary word avoids $\acc{f,f^R,SQ_3}$
if and only if it contains only the squares $\texttt{00}$, $\texttt{11}$, and $\texttt{0101}$ (or $\texttt{00}$, $\texttt{11}$, and $\texttt{1010}$).
Thus, the previous $50$-uniform morphism that avoids $\acc{ABAB.BABA,SQ_3}$ also avoids $\acc{f,f^R,SQ_3}$.

\item $f=AAA$ is self-reverse. This $32$-uniform morphism avoids $\acc{f,SQ_4}$:
{\small
$$
\begin{array}{c}
\texttt{0}\mapsto\texttt{00101001101101001011001001101011}\\
\texttt{1}\mapsto\texttt{00101001101100101101001001101011}\\
\texttt{2}\mapsto\texttt{00100101101001001101101001011011}\\
\end{array}
$$
}
Notice that $\acc{f,SQ_3}$ is 2-unavoidable.

\item $f=ABA.BAAB.BAB$ is self-reverse. This $10$-uniform morphism avoids $\acc{f,SQ_3}$:
{\small
$$
\begin{array}{c}
\texttt{0}\mapsto\texttt{0001110101}\\
\texttt{1}\mapsto\texttt{0001011101}\\
\texttt{2}\mapsto\texttt{0001010111}\\
\end{array}
$$
}

\item $f=AABA.ABAA.BAB$ is self-reverse. This $57$-uniform morphism avoids $\acc{f,SQ_6}$:
{\small
$$
\begin{array}{c}
\texttt{0}\mapsto\texttt{000101011100010110010101100010111001011000101011100101011}\\
\texttt{1}\mapsto\texttt{000101011100010110010101100010101110010110001011100101011}\\
\texttt{2}\mapsto\texttt{000101011100010110010101100010101110010101100010111001011}\\
\end{array}
$$
}
Notice that $\acc{f,SQ_5}$ is 2-unavoidable.

\item $f=AABA.ABAA.BAAB$ is self-reverse. This $30$-uniform morphism avoids $\acc{f,SQ_3}$:
{\small
$$
\begin{array}{c}
\texttt{0}\mapsto\texttt{000101110001110101000101011101}\\
\texttt{1}\mapsto\texttt{000101110001110100010101110101}\\
\texttt{2}\mapsto\texttt{000101110001010111010100011101}\\
\end{array}
$$
}

\item $f=ABAAB$. This $10$-uniform morphism avoids $\acc{f,f^R,SQ_3}$, see~\cite{Ochem2004}:
{\small
$$
\begin{array}{c}
\texttt{0}\mapsto\texttt{0001110101}\\
\texttt{1}\mapsto\texttt{0000111101}\\
\texttt{2}\mapsto\texttt{0000101111}\\
\end{array}
$$
}

\item $f=BAB.ABA.AABB$ is self-reverse. This $16$-uniform morphism avoids $\acc{f,SQ_5}$:
{\small
$$
\begin{array}{c}
\texttt{0}\mapsto\texttt{0101110111011101}\\
\texttt{1}\mapsto\texttt{0100010111010001}\\
\texttt{2}\mapsto\texttt{0001010111010100}\\
\end{array}
$$
}
Notice that $\acc{f,SQ_4}$ is 2-unavoidable.

\item $f=AAB.ABA.ABBA$ is avoided with its reverse. This $84$-uniform morphism avoids $\acc{f,f^R,SQ_5}$:
{\small
$$
\begin{array}{l}
\texttt{0}\mapsto\texttt{000100010111000111010001000101110111010001011100011101000101110111}\\
\texttt{010001110001011101}\\
\texttt{1}\mapsto\texttt{000100010111000111010001000101110100011100010111011101000101110001}\\
\texttt{110100010111011101}\\
\texttt{2}\mapsto\texttt{000100010111000111010001000101110100011100010111010001000101110001}\\
\texttt{110100010111011101}\\
\end{array}
$$
}
Notice that $\acc{f,SQ_4}$ is 2-unavoidable.

\item $f=BAA.ABA.AABB$. This $304$-uniform morphism avoids $\acc{f,SQ_7}$:
% is unavoidable with its reverse, $304$-uniform morphism, avoids~$SQ_{7}$:
{\small
$$
\begin{array}{l}
\texttt{0}\mapsto\texttt{000110001100111000111001100011001110011100110001100011001110011000}\\
\texttt{1110001100111001110011000110011100011100110001100011001110011000111000}\\
\texttt{1100111001110011000110001100111000111001100011001110011100110001110001}\\
\texttt{1001110011000110001100111001110011000110011100011100110001100011001110}\\
\texttt{0111001100011100011001110011}\\
\texttt{1}\mapsto\texttt{000110001100111000111001100011001110011100110001100011001110011000}\\
\texttt{1110001100111001110011000110011100011100110001100011001110011000111000}\\
\texttt{1100111001110011000110001100111000111001100011001110011100110001100011}\\
\texttt{0011100110001110001100111001110011000110011100011100110001100011001110}\\
\texttt{0111001100011100011001110011}\\
\texttt{2}\mapsto\texttt{000110001100111000111001100011001110011100110001100011001110011000}\\
\texttt{1110001100111001110011000110001100111000111001100011001110011100110001}\\
\texttt{1100011001110011000110001100111001110011000110011100011100110001100011}\\
\texttt{0011100110001110001100111001110011000110001100111000111001100011001110}\\
\texttt{0111001100011100011001110011}\\
\end{array}
$$
}
Using the morphism $g_w$ below and the technique in~\cite{BO15},
we can show that $g_w(b_3)$ essentially avoids $\acc{f,SQ_6}$:
{\small
$$
\begin{array}{l}
g_w(\texttt{0})=\texttt{011100111001110001100111001100011000110}\\
g_w(\texttt{1})=\texttt{011100111001100011000110}\\
g_w(\texttt{2})=\texttt{001110011000110}\\ 
\end{array}
$$
}
Notice that $\acc{f,f^R}$ is 2-unavoidable and $\acc{f,SQ_5}$ is 2-unavoidable.

\item $f=BBA.AABA.AABB$. This $160$-uniform morphism avoids $\acc{f,f^R,SQ_{21}}$:
{\small
$$
\begin{array}{l}
\texttt{0}\mapsto\texttt{000101100101110001011100101100010111000101100101110010110001011100}\\
\texttt{1011000101100101110010110001011100010110010111000101110010110001011001}\\
\texttt{011100101100010111001011}\\
\texttt{1}\mapsto\texttt{000101100101110001011100101100010111000101100101110010110001011100}\\
\texttt{1011000101100101110001011001011100101100010111000101110010110001011001}\\
\texttt{011100101100010111001011}\\
\texttt{2}\mapsto\texttt{000101100101110001011001011100101100010111000101100101110001011100}\\
\texttt{1011000101100101110010110001011100010111001011000101100101110001011001}\\
\texttt{011100101100010111001011}\\
\end{array}
$$
}
This $202$-uniform morphism avoids $\acc{f,SQ_5}$:
{\small
$$
\begin{array}{l}
\texttt{0}\mapsto\texttt{000110100111011010001101010001110110100110110101000111011010001101}\\
\texttt{0100011101101010001101001110110100110110101000110100111011010100011101}\\
\texttt{101000110101000111011010100011010011101101010001110110100110110101}\\
\texttt{1}\mapsto\texttt{000110100111011010001101010001110110100110110101000110100111011010}\\
\texttt{1000111011010001101010001110110101000110100111011010100011101101001101}\\
\texttt{101010001101001110110100110110101000111011010001101010001110110101}\\
\texttt{2}\mapsto\texttt{000110100111011010001101010001110110100110110101000110100111011010}\\
\texttt{1000111011010001101010001110110101000110100111011010011011010100011101}\\
\texttt{101000110101000111011010100011010011101101010001110110100110110101}\\
\end{array}
$$
}
Notice that $\acc{f,f^R,SQ_{20}}$ is 2-unavoidable and $\acc{f,SQ_4}$ is 2-unavoidable.
\end{itemize}

We start by checking that every morphism is synchronizing, that is, for every letters $a,b,c\in\Sigma_3$,
the factor $g(a)$ only appears as a prefix or a suffix in $g(bc)$.

For every $q$-morphism $g$, the sqf-$g$-images are claimed to avoid $SQ_t$ with $2t<q$.
Let us prove that $SQ_t$ is avoided.
We check exhaustively that the sqf-$g$-images contain no square $uu$ such that $t\le|u|\le2q-2$.
Now suppose for contradiction that an sqf-$g$-image contains a square $uu$ with $|u|\ge 2q-1$.
The condition $|u|\ge 2q-1$ implies that $u$ contains a factor $g(a)$ with $a\in\Sigma_3$.
This factor $g(a)$ only appears as the $g$-image of the letter $a$ because $g$ is synchronizing.
Thus the distance between any two factors $u$ in an sqf-$g$-image is a multiple of $q$.
Since $uu$ is a factor of an sqf-$g$-image, we have $q\ |\ |u|$.
Also, the center of the square $uu$ cannot lie between the $g$-images of two consecutive letters,
since otherwise there would be a square in the pre-image.
The only remaining possibility is that the ternary square free word contains a factor $aXbXc$
with $a,b,c\in\Sigma_3$ and $X\in\Sigma_3^+$ such that $g(aXbXc)=bsYpsYpe$ contains the square $uu=sYpsYp$,
where $g(X)=Y$, $g(a)=bs$, $g(b)=ps$, $g(c)=pe$.
Then, we also have $a\ne b$ and $b\ne c$ since $aXbXc$ is square free.
Then $abc$ is square free and $g(abc)=bspspe$ contains a square with period $|s|+|p|=|g(a)|=q$.
This is a contradiction since the sqf-$g$-images contain no square with period $q$.

Let us show that for every formula $f$ above and corresponding morphism $g$, $g$ avoids~$f$.
Notice that $f$ is not square free, since the only avoidable square free binary formula is $ABA.BAB$,
which is not 2-avoidable. We distinguish two kinds of formula.

A formula is \emph{easy} if every appearing variable is contained in at least one square.
Every potential occurrence of an easy formula then satisfies $|A|<t$ and $|B|<t$ since $SQ_t$ is avoided.
The longest fragment of every easy formula has length $4$.
So, to check that $g$ avoids an easy formula, it is sufficient to consider
the set of factors of the sqf-$g$-images with length at most $4(t-1)$.

A formula is \emph{tough} if one of the variables is not contained in any square.
The tough formulas have been named so that this variable is $B$.
The tough formulas are $ABA.BAAB.BAB$, $ABAAB$, $AABA.ABAA.BAAB$, and $AABA.ABAA.BAB$.
As before, every potential occurrence of a tough formula satisfies $|A|<t$ since $SQ_t$ is avoided.
Suppose for contradiction that $|B|\ge2q-1$. By previous discussion, the distance between any
two occurrences of $B$ in an sqf-$g$-image is a multiple of $q$.
The case of $ABA.BAAB.BAB$ can be settled as follows.
The factor $BAAB$ implies that $q$ divides $|BAA|$ and the factor $BAB$ implies that $q$ divides $|BA|$.
This implies that $q$ divides $|A|$, which contradicts $|A|<t$.
For the other formulas, only one fragment contains $B$ twice. This fragment is said to be \emph{important}.
Since $|A|<t$, the important fragment is a repetition which is ``almost'' a square.
The important fragment is $\boldsymbol{B}A\boldsymbol{B}$ for $AABA.ABAA.BAB$,
$\boldsymbol{B}AA\boldsymbol{B}$ for $AABA.ABAA.BAAB$, and $\boldsymbol{AB}A\boldsymbol{AB}$ for $ABAAB$.
Informally, this almost square implies a factor $aXbXc$ in the ternary pre-image, such that $|a|=|c|=1$ and $1\le|b|\le2$.
If $|X|$ is small, then $|B|$ is small and we check exhaustively that there exists no small occurrence of $f$.
If $|X|$ is large, there would exist a ternary square free factor $aYbYc$ with $|Y|$ small, such that $g(aYbYc)$
contains the important fragment of an occurrence of $f$ if and only if $g(aXbXc)$
contains the important fragment of a smaller occurrence of $f$.

% gcc -std=c99 -Wunused -W -Wall -Wextra -Werror -O3 -o alg alg.c -lgmp
% 
% \section{Some ternary formulas}\label{sec:ter}
% In this section, we generalize Thue's results $\lambda(AA)=3$ and $\lambda(AAA)=2$
% by providing a formula that is harder to avoid but with the same avoidability index.

\section{Concluding remarks}\label{sec:con}

% When a formula is avoided by exponentially many binary words, there is room to show more than just 2-avoidability.
% A first natural idea is to get a lower bound on the growth rate of binary words avoiding the formula.
% The lower bound provided by a $k$-uniform morphism is $\alpha^{1/k}$ where $\alpha=1.30125\dots$ is the growth rate
% of ternary square free words. The bound is thus better when $k$ is smaller. For example, it can be shown that the formula
% $ABAB.BABA$ is avoided by the image of every ternary square free word by the (non-synchronizing) morphism
% $\texttt{0}\mapsto\texttt{00}$, $\texttt{1}\mapsto\texttt{01}$, $\texttt{2}\mapsto\texttt{11}$, which gives the bound $\alpha^{1/2}=1.14072\dots$
% However, the bounds obtained via morphisms are very bad compared to other methods~\cite{KR:2011}.

From our results, every minimally 2-avoidable binary formula, and thus every 2-avoidable binary formula,
is avoided by some morphic image of $b_3$.

What can we forbid so that there exists only polynomially many avoiding words ?
The known examples from the literature~\cite{BO15,BNT89,Thue06} are: 
\begin{itemize}
 \item one pattern and two factors:
 \begin{itemize}
  \item $b_3$ essentially avoids $AA$, $\texttt{010}$, and $\texttt{212}$.
  \item A morphic image of $b_5$ essentially avoids $AA$, $\texttt{010}$, and $\texttt{020}$.
  \item A morphic image of $b_5$ essentially avoids $AA$, $\texttt{121}$, and $\texttt{212}$.
  \item $b_2$ essentially avoids $ABABA$, $\texttt{000}$, and $\texttt{111}$.
 \end{itemize}
 \item two patterns: $b_2$ essentially avoids $ABABA$ and $AAA$.
 \item one formula over three variables: $b_4$ and two words obtained from $b_4$ by letter permutation essentially avoid $AB.AC.BA.BC.CA$.
\end{itemize}
Now we can extend this list:
\begin{itemize}
 \item one formula over two variables:
 \begin{itemize}
  \item $g_x(b_3)$ essentially avoids $AAB.BAA.BBAB$.
  \item $\acc{g_x(b_3), g_t(b_3)}$ essentially avoids $ABA.AABB$ (or $BBA.ABA.AABB$, or $AABA.AABB$).
  \item $\acc{g_x(b_3), g_y(b_3), g_z(b_3), g_{\overline{z}}(b_3)}$ essentially avoids $AA.ABA.ABBA$.
 \end{itemize}
 \item one pattern over three variables: $ABACAABB$ (same as $ABA.AABB$) or $AABACAABB$ (same as $AABA.AABB$).
\end{itemize}

% Although the formula ? does not divide ? and ?, the only bi-infinite word avoiding $AABA.ABB.BBA$ also avoids $AA.ABA.ABBA$ and $ABA.AABB$ and others.

% \bibliography{biblio}

\end{document}